\documentclass[%
groupedaddress,
nofootinbib,
amsmath,amssymb,
aps,
prd
]{revtex4-2}

\usepackage{amsthm}
\usepackage{xspace,enumitem,bigints}
\usepackage[bookmarksnumbered,colorlinks]{hyperref}
\usepackage[english]{babel}
\usepackage[T1]{fontenc}
\usepackage{color}
\usepackage{mathrsfs}  % Load the mathrsfs package for \mathscr command

\newtheorem{theorem}{Theorem}[section]

\newtheorem{proposition}[theorem]{Proposition}
\theoremstyle{definition}
\newtheorem{definition}[theorem]{Definition}
\newtheorem{remark}[theorem]{Remark}

\theoremstyle{definition}

\newcommand{\bt}{\begin{theorem}}                     
	\newcommand{\et}{\end{theorem}}                       
\newcommand{\bd}{\begin{definition}}                  
	\newcommand{\ed}{\end{definition}}                    
\newcommand{\bl}{\begin{lemma}}                       
	\newcommand{\el}{\end{lemma}}                                   
\newcommand{\bpr}{\begin{proposition}}                  
	\newcommand{\epr}{\end{proposition}}                    
\newcommand{\bere}{\begin{remark}}                      
	\newcommand{\ere}{\end{remark}}                         

\newcommand{\beq}{\begin{equation}}
	\newcommand{\eeq}{\end{equation}}
\def\bal#1\eal{\begin{align}#1\end{align}}              
\def\baln#1\ealn{\begin{align*}#1\end{align*}}          
\def\bml#1\eml{\begin{multline}#1\end{multline}}        
\def\bmln#1\emln{\begin{multline*}#1\end{multline*}}  
\def\bga#1\ega{\begin{gather}#1\end{gather}}
\def\bgan#1\egan{\begin{gather*}#1\end{gather*}}

\newcommand{\de}{\mathrm{d}}                        
                     
\newcommand{\N}{\ensuremath{\mathbb{N}}\xspace}     
\newcommand{\R}{\ensuremath{\mathbb{R}}\xspace}     
\newcommand{\Sp}{\ensuremath{\mathbb{S}}\xspace}     
     
\newcommand{\eps}{\varepsilon}

\newcommand{\calh}{\ensuremath{\mathcal{H}}\xspace}

\newcommand{\HH}{{\mathcal{H}}}
\def\bb#1\eb{\textcolor{blue}
	{#1}} %
\def\br#1\er{\textcolor{red}
	{#1}} %
\def\bv#1\ev{\textcolor{green}
	{#1}} %
\def\bm#1\em{\textcolor{magenta}
	{#1}}
\usepackage[normalem]{ulem}
\begin{document}
	
	\title{Lorentzian-Euclidean black holes and Lorentzian to Riemannian metric transitions}
	
	\author{Rossella Bartolo}
	\email{rossella.bartolo@poliba.it}
	\affiliation{Dipartimento di Meccanica, Matematica e Management, Politecnico di Bari, Bari, Italy}
	
	\author{Erasmo Caponio}
	\email{erasmo.caponio@poliba.it}
	\affiliation{Dipartimento di Meccanica, Matematica e Management, Politecnico di Bari, Bari, Italy}
	
	\author{Anna Valeria Germinario}
	\email{anna.germinario@uniba.it}
	\affiliation{Dipartimento di Matematica, Università degli Studi di Bari Aldo Moro, Bari, Italy}
	
	\author{Miguel Sánchez}
	\email{sanchezm@ugr.es}
	\affiliation{Departamento de Geometría y Topología, Facultad de Ciencias, \& IMAG (Centro de Excelencia María de Maeztu) Universidad de Granada, Granada, Spain}
	
	\begin{abstract}
		In recent papers on spacetimes with a signature-changing metric,  Capozziello et al. \cite{CaDeBa24} and  Hasse and Rieger \cite{HasRie24b} introduced, respectively, the concept of a Lorentzian-Euclidean black hole and new elements   for Lorentzian-Riemannian signature change. In both cases the transition in the signature happens on a hypersurface $\mathcal{H}$. The former is a signature-changing modification of the Schwarzschild spacetime satisfying the vacuum Einstein equations in a weak sense. Here $\mathcal{H}$ is the event horizon which serves as a boundary beyond which time becomes imaginary. We clarify an issue appearing in \cite{CaDeBa24} based on numerical computations which suggested that an observer in radial free fall would require an infinite amount of proper time to reach the event horizon. We demonstrate that the proper time needed to reach the horizon remains finite, consistently with  the classical Schwarzschild solution,  and suggesting that the model in \cite{CaDeBa24} should be  revised.   About the latter, we stress that $\mathcal{H}$ is naturally a spacelike hypersurface   related to  the future or past causal boundary of the Lorentzian  sector.
		Moreover, a number of geometric interpretations appear, as the degeneracy
		of the metric $g$ corresponds to the collapse of the causal cones into a line, the
		degeneracy of the dual metric $g^*$ corresponds to collapsing into a
		hyperplane, and additional geometric structures  on $\mathcal{H}$ (Galilean
		and dual Galilean) might be explored.
		\end{abstract}
	
	\maketitle
	
	\section{Introduction}
	\label{sec:intro}
	In quantum theories of gravity, various approaches consider different treatments of the metric signature \cite{Haywar92a, GibHar90}. In particular, this is exemplified in quantum cosmology by the Hartle-Hawking no-boundary proposal \cite{HarHaw83}, which incorporates a transition from a Riemannian to a Lorentzian metric. Such transitions in  metric signature have subsequently  been studied also in classical spacetimes from the perspective of the junction conditions that they must satisfy (see, e.g., \cite{CarEll95} and the references therein).	Moreover, the local  and the global geometry of a manifold endowed with a signature-changing $(0,2)$ bilinear tensor field have been studied respectively in \cite{KosKri93} and in \cite{KosKri97}.
	
	Recent works by Capozziello et al. \cite{CaDeBa24} and Hasse and Rieger \cite{HasRie24b} have explored different aspects of signature changes. While both papers deal with signature transitions, they approach the subject from distinct perspectives and with different physical interpretations. Capozziello et al. introduce a novel type of black hole solution where the signature changes across the event horizon, while	
	Hasse and Rieger develop a geometric framework for Lorentzian-Riemannian transitions. 
	 In this paper, we aim to clarify some key aspects of both approaches. Particularly, we focus on the nature of the transition hypersurface \calh in both frameworks. 
	 Indeed, in  \cite{CaDeBa24}, the event horizon has a lightlike nature and the issue is whether some freely falling observers arrive at \calh in infinite proper time, differently from the classical Schwarzschild spacetime. We  prove that this does not happen and neither occurs if the jump in the coordinates between the outer and inner parts is smoothened by using a sign changing function. Indeed, this is checked first in  Schwarzschild coordinates and, then,  in  Gullstrand-Painlevé coordinates of the  Lorentzian-Euclidean black hole.  The latter proof appears to conflict with numerical data  pointed out in \cite{CaDeBa24}   and demands a full revision of the model by their proponents.  
	 
Regarding \cite{HasRie24b}, it is  rooted in the Hartle-Hawking no-boundary proposal, which suggests  that a Riemannian component  eliminates the need for initial boundary conditions at the Big Bang. While observers would not be expected to travel into the (essentially quantum) Riemannian region, the transition hypersurface can nevertheless be analyzed from the spacetime perspective as the causal boundary of spacetime. In the last section, we introduce several concepts for this analysis.

	\section{Lorentzian-Euclidean black holes and finiteness of arrival proper time}
	\label{sec:blackholes}
	
	\subsection{The signature-changing metric}
	In   \cite{CaDeBa24},  the following signature-changing metric on $M:=\R\times (0,\infty)\times \Sp^2$ in spherical  coordinates $(t, r, \theta, \varphi)$ is introduced: 
	\beq\label{les}
	g = 
	- \varepsilon (r)\left( 1 -\frac{2m}{r} \right) \de t^2+\left(1 -\dfrac{2m}{r}\right)^{-1}\de r^2+
	r^2\de\Omega^2 ,
	\eeq
	where $d\Omega^2 = \sin^2\theta d\varphi^2 + d\theta^2$ is the standard metric
	of the unit $2$--sphere $\Sp^2$ in $\R^3$, $m>0$ is a positive parameter
	and
	\beq\label{sign}
	\varepsilon (r) = {\rm sign}\left(1 -\frac{2m}{r}\right)\!\!,\text{ for $r>0$, $r\neq 2m$,}\quad \eps(2m)=0.
	\eeq
	The metric signature switches from the usual Lorentzian one on the region  $V^+:=\R\times (2m, \infty)\times \Sp^2$ to a semi-Riemannian one of index $2$ on  $V^-:=\R\times (0,2m)\times \Sp^2$, upon crossing the  hypersurface  $\calh = \{  r=2m\}$.   In particular,  
	for fixed $\theta$ and $\varphi$ the metric  on $\R\times (0,2m)$ is Riemannian, up to a negative sign, and for this reason in analogy with the so-called  Euclidean-Schwarzschild metric (see \cite{Allen84}), the authors    call it a  {\em Lorentzian-Euclidean Schwarzschild metric} (hereafter {\em LES metric}).  The line element \eqref{les}, away from \calh, corresponds to the Schwarzschild metric with an imaginary time substitution in the region $V^-$: the coordinate time $t$ is replaced by $i t$ when the event horizon \calh is crossed. 
	
	The LES metric  \eqref{les} is both divergent and degenerate on $\calh$. However, the use of Gullstrand-Painlevé coordinates eliminates the divergence, leaving only the degeneracy to be addressed. In these coordinates $(\mathscr T, r, \theta, \varphi)$,  \eqref{les} becomes:
	\beq\label{gp}
	g=- \varepsilon(r) \de \mathscr T^2 + \left(\de r +\sqrt{\varepsilon(r)}\sqrt{\frac{2m}{r}} \de \mathscr T\right)^2+
	r^2\de\Omega^2.
	\eeq
	In \cite{CaDeBa24} the authors consider a smooth approximation of the function $\eps$, replacing \eqref{sign} by
	\beq\label{app}
	\varepsilon_{\rho,\kappa}(r)=\dfrac{(r-2m)^{\frac{1}{2\kappa + 1}}}{[(r-2m)^2 + \rho]^\frac{1}{2(2\kappa + 1)}},
	\eeq
	with $\rho>0$ and $\kappa\in\N$. Notably, $\varepsilon_{\rho,\kappa}\to \eps$ in $C^\infty$ norm in any set of the type $(0, 2m-a]\cup [2m+a, \infty)$, $a>0$, both as $\rho \to 0$ and $\kappa\to \infty$.  After some elementary algebraic manipulations, \eqref{gp} can be written as 
	\beq\label{gps}
	g_{\varepsilon_{\rho,\kappa}}:= - \varepsilon_{\rho,\kappa}(r)\left(1-\frac{2m}{r}\right) \de \mathscr T^2+\de r^2 + 2\sqrt{\varepsilon_{\rho, \kappa}(r)}\sqrt{\frac{2m}{r}} \de \mathscr T \de r + r^2\de \Omega^2.
	\eeq
	This metric is continuous everywhere  but, since $\varepsilon_{\rho,\kappa}(2m)=0$, is still not differentiable at \calh,  while it is   smooth both on $V^+$ and $V^-$ (where it is complex).  The metrics \eqref{gps} are used  in  a  Hadamard regularization type argument (see \cite[\S 9.6.2]{PoiWil14}), leading  the authors of  \cite{CaDeBa24} to conclude  that \eqref{gp} is a  well-defined signature-changing solution of the vacuum Einstein field equations that they call a {\em Lorentzian-Euclidean black hole}.
	
	\subsection{On  the geodesics of the signature-changing metric} \label{2b}
      Let $\tilde \varepsilon: [0,\infty)\rightarrow (-1,1)$ be  a continuous function which is smooth on $[0,\infty)\setminus\{2m\}$,  such that $\tilde \eps'(r)>0$, $\tilde \varepsilon(0)<0$, $\tilde \eps(r)\to 1$ as $r\to \infty$ and $\tilde \varepsilon(2m)=0$, $\tilde \varepsilon'(2m)=+\infty$ (so $\tilde{\eps}$ behaves as any $\varepsilon_{\rho,\kappa}$ in \eqref{app}).  Let $g_{\tilde\varepsilon}$ be defined as in \eqref{les} with $\tilde\varepsilon$ replacing $\eps$. For the sake of completeness and  just  to reproduce formulas appearing in \cite{CaDeBa24}, we remark that,     in Eqs. \eqref{E}--\eqref{propertime} below, $\tilde \eps$  can be taken equal to   the step function $\eps$ in \eqref{sign}.  (In any case,  the reader can check that  if $\tilde \eps'\geq 0$ then in Proposition~\ref{concave} we obtain that $r$ is concave and Proposition~\ref{fpt} still holds.)
	 
	 In \cite{CaDeBa24}, following \cite{ElSuCH92, Ellis92}\footnote{In \cite{Ellis92}, the changing-signature metric takes the form of a Robertson-Walker metric with a lapse function given by the step function $\varepsilon$. While that work also considers a continuous representation obtained through a coordinate change with a sign-changing function (such as $\tilde{\varepsilon}$), its setting fundamentally differs from \cite{CaDeBa24} in that there is no divergent singularity at the signature-changing hypersurface.}, a family of   continuous   privileged curves $\gamma:I\to M$, $\gamma=\gamma(\tau)=\big(t(\tau), r(\tau), \theta(\tau), \varphi(\tau)\big)$,    is introduced: each curve $\gamma$ is   smooth off the instants where it crosses \calh,  it satisfies  
	\beq   g_{\tilde\eps} (\dot\gamma, \dot\gamma)=-\eps(r)\label{privileged}\eeq
	(see \cite[Eq. (6)]{CaDeBa24}) and  is a critical  curve of the Lagrangian  
	\[-\eps(r)g_{\tilde\eps},  \]
	  at least separately on $V^+$ and $V^-$.
%	 when restricted to each of the   region $V^+$ and $V^-$, see \cite[\S 7]{Ellis92} (so these restrictions are actually geodesics of the Lorentzian metric  $g_{\tilde\eps}|_{V^+}$ and the semi-Riemannian one $g_{\tilde\eps}|_{V^-}$).
	  Since this Lagrangian is independent of $t$,  there exists a    piecewise   constant function   $E$ along $\gamma$ such that (see \cite[Eq. (8)]{CaDeBa24})
	%$\varepsilon (r)\left( 1 -\frac{2m}{r} \right)=|\Lambda (r) |$.
	% Let   $\gamma:I\to M$,  $\gamma(\tau)=(t(\tau),r(\tau))$ be an absolutely continuous radial curve, and let us assume that $\gamma$ crosses \calh. 
	%  and it is a critical curve of the energy functional of \eqref{les}. Then by the fundamental lemma of calculus of variations we get  	that there exists a constant $E$ s.t. 
	\beq\label{E}E  =   \tilde \eps(r) \eps(r)  \Lambda(r)\dot{t},\eeq
	  where  	$\Lambda(r):= \left( 1 -\frac{2m}{r} \right)$. 
	In \cite{CaDeBa24}, the curves $\gamma$ are timelike in the Lorentzian sector   where then  $E\neq 0$. 
	%(and, from \eqref{E}, necessarily $\dot t=\infty$ at the instants where $\gamma$ crosses \calh).  
	Moreover,  from \eqref{privileged}, assuming that   $\gamma(\tau)=(t(\tau),r(\tau))$ is radial, we get:
	\[-\eps(r) =  -  \tilde \eps(r)  \Lambda (r)\dot{t}^2 + \frac{\dot{r}^2}{\Lambda(r)}.\]
	Hence, the radial component $r=r(\tau)$ satisfies (see \cite[Eq.(39)]{CaDeBa24})  		
	\beq\label{dottr}\dot r^2=\frac{E^2}{  \tilde \eps(r) \eps^2(r) }-\eps(r)\Lambda (r).
	\eeq
	Thus,    taking into account that $\tilde \eps\eps^2=\tilde \eps$     the proper time that a radial infalling  observer   takes to reach \calh is equal to
	\beq\label{propertime}
	\tau = \int_{2m}^{r_0} \frac{\sqrt{  \tilde \eps(r) }}{\sqrt{E^2-  \tilde \eps(r)\eps(r)  \Lambda(r)}}\de r,\quad r_0>2m,
	\eeq
  which is clearly finite, since $E\neq 0$   on $V^+$  as recalled above.  
%	which coincides with the proper time  in classical  Schwarzschild spacetime $V^+$, where $\eps\equiv1$, and  then is finite.  	
	 We guess that the authors in \cite{CaDeBa24}   view   the singularity   at \calh    of the signature-changing metric  to be problematic for analyzing the proper time of infalling observers  as discussed above.    However,  \eqref{propertime} can be also derived in Gullstrand-Painlev\'e coordinates   (see Appendix \ref{appA}) where the signature-changing metric has no diverging singularity as recalled above.\footnote{Although proper time is preserved under coordinate transformations, the 
	 	smoothing procedure in Gullstrand-Painlevé coordinates introduce a further complication due to the complex valued function $\sqrt{\tilde \eps}$ which does not appear in
	 	 Schwarzschild coordinates.
	 	Anyway, the computations in Appendix~\ref{appA} confirm that this does not affect
	 	the finiteness of proper time for radial infalling  observers.} 
	 %A fundamental point in these computations is the assumption in \cite{CaDeBa24} that there exist a family  of    privileged observers crossing the signature transition hypersurface and parametrized with a continuous parameter $\tau$, as in \cite{Ellis92}.
	    We  must note   that while  \eqref{propertime}  is compatible with  the continuity   of  the parameter $\tau$  at \calh, it yields  imaginary values for $r_0<2m$.     Indeed, while in \cite{CaDeBa24}   $\tau$ is used   to adapt Gullstrand-Painlevé coordinates to the manifold $M$ with the signature-changing metric \eqref{les}, the analysis of radial infalling observers   relies  instead on    Eq. (46) of \cite{CaDeBa24},   leading the authors to conclude   that infinite proper time is required  to reach the event horizon.   They support this conclusion through     two approaches: first, by approximating $\eps_{\rho,\kappa}$ in \eqref{app} with a constant $\eps$ and taking the limit as   $\eps \to 0$,   and second, by numerically solving \cite[Eq. (46)]{CaDeBa24} for  the proper time of a radial geodesic within the spacetime region $(V^+, g_{\eps_{\rho, \kappa}})$,   using an  unspecified  $\eps_{\rho, \kappa}$ close to $\eps$ (see \cite[\S IV.A]{CaDeBa24}).

	   Following then \cite{CaDeBa24}'s framework, in the next section we  analyze  timelike  geodesics on $V^+$    using   the family of approximating metrics \eqref{gps} in Gullstrand-Painlevé coordinates.   We demonstrate   that the proper time of an infalling observer remains finite    because the radial coordinate of an infalling geodesic is strictly  concave  and therefore cannot have the asymptote $r=2m$ while staying  above   the horizon. 
	 
	\subsection{Finiteness of the proper arrival time at the horizon in the Gullstrand-Painlev\'e  coordinates}
	  Let now $g_{\tilde\varepsilon}$ be  as in \eqref{gps} with $\tilde\varepsilon$, defined as in previous subsection,  replacing $\varepsilon_{\rho,\kappa}$.    	In what follows we set
	\begin{equation} \label{bgp}
		\beta(r)=\tilde\varepsilon(r)\left(1-\frac{2m}{r}\right).
	\end{equation}  
	\bpr\label{concave}
	Let $\gamma (s) =  \big(\mathscr T(s), r(s), \theta(s), \varphi(s)\big)$, $s\in I$,  be a causal %future pointing 
	geodesic with respect to the metric $g_{\tilde\varepsilon}$, such that 
	$r(s) \in ( 2m, 3m) $, for all $s\in I$. Then function  $r$ is strictly concave.
	\epr
	\begin{proof} By   \eqref{rtt}-\eqref{rphph},  the radial component $r$ of a  geodesic $\gamma$  for $g_{\tilde\eps}$  
		verifies in $I$
		\begin{align}
			\ddot r & = -\Gamma^r_{\mathscr T\!\!\mathscr T}\dot{\mathscr T}^2  - \Gamma^r_{rr}\dot r^2  -2\Gamma^r_{\mathscr T\!r}\dot r \dot{\mathscr T} 	-\Gamma^r_{\theta\theta}\dot\theta^2 - \Gamma^r_{\varphi\varphi}\dot\varphi^2   \nonumber \\
			&  =  -\frac{m}{r^2}\beta \dot{\mathscr T}^2-\frac 12\beta^2 \frac{\tilde \varepsilon'}{\tilde \varepsilon^2}\dot{\mathscr T}^2     + \frac{m}{r^2}\dot r^2 - \frac{m}{r}\frac{\tilde \varepsilon'}{\tilde \varepsilon}\dot r^2     
			+ \frac{2m}{r^2} \sqrt{\frac{2m}{r}}\sqrt{\tilde\varepsilon}\dot{\mathscr T}\dot r +  \beta  \frac{\tilde \varepsilon'}{\tilde \varepsilon^2} \sqrt{\frac{2m}{r}}\sqrt{\tilde\varepsilon}\dot{\mathscr T}\dot r  
			+ (r-2m)\Omega,  \label{ddot}
		\end{align} 
		where $\Omega  = \sin^2\theta \dot \varphi^2 + \dot \theta^2 $. 
		Since  $\gamma$ is  causal then
		\beq\label{en}
		2\sqrt{\tilde \varepsilon}\sqrt{\frac{2m}{r}}\dot{\mathscr T}\dot r\leq \beta\dot t^2 - \dot r^2 - r^2\Omega.
		\eeq
		By \eqref{ddot} we get
		\beq \label{en1}
		\ddot{r}	\leq - \frac{m}{r}\frac{\tilde \varepsilon'}{\tilde \varepsilon}\dot r^2 - \frac{\beta}{2}\frac{\tilde \varepsilon'}{\tilde\varepsilon^2}\dot r^2 -\frac{\beta}{2}\frac{\tilde\varepsilon'}{\tilde\varepsilon^2}r^2\Omega + (r-3m)\Omega .
		\eeq
 Hence,  it follows that  $\ddot r (s) \leq 0$  for any $s \in I$. Taking into account  that $\tilde\eps'>0$,  if $\dot r(s)  \not=0$ or $\Omega(s)\neq 0$,  then $\ddot r (s) < 0$. If   $\Omega (s)  =0$ and $\dot r(s)  = 0$,   then $\dot t (s) \not =0$ (recall that $\gamma$  is causal). So by \eqref{ddot} 
$\ddot r (s) < 0$ as well.

%		Hence,  it follows that  $\ddot r (s) \leq 0$  for any $s \in I$. Taking into account  that $\tilde\eps'>0$,  if $\dot r(s)  \not=0$ or $\Omega(s)\neq 0$  then $\ddot r (s) < 0$. If   $\Omega (s)  =0$ and $\dot r(s)  = 0$,  as 
%		$$-\beta \dot{\mathscr T} + \sqrt{\tilde \eps} \sqrt{\frac{2m}{r}} \dot r = c \not= 0$$
%		(recall that $\gamma$  is causal) then $\dot{\mathscr T}(s) \not =0$ and by \eqref{ddot} 
%		$\ddot r (s) < 0$ as well.
	\end{proof}

	Since $r=r(s)$ is strictly concave,  any causal geodesic arc in the region 
	\[S:=\R\times (2m, 3m)\times\Sp^2,\]
	cannot have a minimum in the interior of its  interval of parametrization. (In particular, no causal spatially closed orbit exists in that region.)   This implies that,  if  $r^* \in (2m,3m)$, the boundary of the region $V^+_{r^*}:= \R\times (r^*, \infty)\times \Sp^2$ is timelike and lightlike   convex in the sense of 
	\cite[Ch. 4]{Masiel94} and
	\cite{BaGeSa02, Caponi13}. 
	
	\bere\label{timeconvex}\em
	We  observe that the above proof works also for the radial component of a timelike  geodesic in the region $S$ endowed with  the Schwarzschild metric in Gullstrand-Painlev\'e coordinates (i.e., for $\tilde\eps\equiv 1$) by using the strict inequality in \eqref{en}. 
	\ere

	A critical issue regarding the Lorentzian-Euclidean black hole is whether observers can reach the horizon in finite proper time. We now demonstrate that  the proper time is finite, as follows naturally from the concavity of the radial component of any causal geodesic.

	Affine transformations of the parameter $s$ do not alter the strict concavity of the radial component  $r=r(s)$   of any geodesic arc  $\gamma: [a,b)\rightarrow S$ in  the metric \eqref{gps}.  Consequently, if $s$ represents the proper  time and $b=\infty$, $r(s)$ would be a strictly concave function with  the asymptote $r=2m$ as $b\to \infty$ and  $r(s)>2m$ for all $s$. However,  this is impossible; therefore $b<\infty$.  This  contradicts what is claimed in \cite{CaDeBa24}, in particular in Fig. 3.  
	
	However, let us examine  the approximation scheme to establish whether $b_{\tilde \eps}\to \infty$ as $\tilde\eps$ tends to $\eps$.  
	Let $(\tilde\varepsilon_n)_n$ be a sequence of functions as the ones introduced above  \eqref{bgp}, and  $\eps$ as in \eqref{sign}. We assume  that:
	\begin{itemize}
		\item[(a)] $\tilde\eps_n(r)\leq \tilde\eps_{n+1}(r)$ for each $n\in \N$ and  $r\in [2m, \infty)$;
		\item[(b)] $\tilde\eps_n\to \eps$ in $C^2((0, 2m-a]\cup [2m+a, \infty), \R)$,  $a>0$;
	\end{itemize}
	for example,  $\tilde\varepsilon_n$ can be taken as in \eqref{app} for a suitable choice of $\rho$ and $\kappa$. Let us  define $\beta_n:=\tilde\varepsilon_n\left(1-\frac{2m}{r}\right).$
	 \begin{proposition}\label{fpt}
	Let $\gamma\colon [a,b)\to S$, $\gamma=\gamma(s)$,	$b\in (a, \infty]$, be a timelike  geodesic of the metric $g$ in \eqref{gp}, parametrized w.r.t. proper time and  such that $\lim_{s\to b}\gamma(s)\in \calh$. Then $b<\infty$.
	\end{proposition}
	\begin{proof}
	 Following the approach in \cite{CaDeBa24}, which suggests that geodesics of $g$ are approximated by geodesics of
	$g_{\tilde \varepsilon_n}$
	at least until they lay in the region  $V^+$, it is reasonable to assume that there exists a sequence    $\gamma_n: [a, b_n)\to S$,  $\gamma_n(s)=(\mathscr T_n(s), r_n(s), \theta_n(s), \varphi_n(s)\big)$ of timelike geodesics of \eqref{gps},  with $\tilde\eps_n$ replacing $\eps_{\rho,\kappa}$,   parametrized with proper time, such that  $\gamma_n(a)=\gamma(a)=p_0:=(\mathscr T_0,r_0,\theta_0,\varphi_0)\in S$,
	and 	the sequence of initial 
	vectors $v_n$ in $T_{p_0} S$ converges to  $v=\dot\gamma(a)$. 
	From Prop. \ref{concave},   we can assume that $r_n$ is strictly decreasing on $[0,b_n]$, and $r(b_n)=2m$ for all $n\in \N$. (Furthermore, as $r_n$ is strictly concave, then  $\dot r_n(s)\not =-\infty$ for all $s\in [0,b_n)$.) 
	Let us also assume that  $b_n\rightarrow b\in (0,\infty]$ as $n\rightarrow  \infty$. So,  $r_n$ is definitively well-defined on $[0,c]$ for any $c\in (0,b)$  
	and, up to a subsequence, still denoted by 
	$r_n$, we can also assume that $r_n(c)\rightarrow \bar r \in (2m,r_0]$ as 
	$n\rightarrow  \infty$. 	Since $\tilde \varepsilon_n\rightarrow 1$ in $C^2([\bar r,3m],\R)$, it follows that 
	$g_{\tilde \varepsilon_n}\rightarrow g$  on  $V^+_{\bar r}$ in the $C^2$-topology,
	where $g$ is the LES metric \eqref{gp}. 
%	Actually $g$ is the   Schwarzschild metric 
%	in Gullstrand-Painlev\'e coordinates on $V^+_{\bar r}$  by \eqref{sign}. 
Consequently,  the sequence  $\gamma_n:[0,c]\to S$ converges in the $C^2$-topology  to  $\gamma|_{[0,c]}$.
In particular, as  (pointwise) limit of concave functions  $r:[0,c]\rightarrow [\bar r, r_0]$  is  concave. Since $c\in (0,b)$ is arbitrary, $r$ is  concave on $[0,b)$. Moreover, $r$ cannot be constant, as on the interval $[0,c]$ it is  the radial component of a timelike geodesic of the Schwarzschild metric and is therefore strictly concave by Remark~\ref{timeconvex}.
	In any case a concave function cannot  admit a horizontal asymptote if its graph  lays above the asymptote.  This demonstrates that $b<\infty$, implying that a freely falling observer in the LES metric requires a finite proper time to reach the event horizon.
	\end{proof}

	\section{Remarks on the Lorentzian-Riemannian signature change}
%	\label{sec:remarks}
	%\footnote{Stilystic issue. Each paragraph finishes in a sentence in italics. Each one might be written as a Remark (and, eventually, some footnotes might be included in the text).}
	
	Hasse and Rieger \cite{HasRie24b} working upon Kossowski and Kriele \cite{KosKri93, KosKri97} consider a 
	transition Lorentz-Riemann in a subset  $\HH$ of a singular semi-Riemannian manifold $M$ where the metric $g$ degenerates, but its differential  does not vanish  (thus,  $\HH$ becomes a smooth hypersurface) and the radical of $g$ is  1-dimensional and transverse to $\HH$. In this case, the metric can be written, at least locally around $\HH$, as 
	\begin{equation}\label{e1}
		g=-tdt^2+g_{ij}(t,x^1, \dots, x^{n-1})dx^idx^j, \qquad \HH=\{t=0\}.
	\end{equation}  
	
	\begin{enumerate}
		\item[$(1)$]
	The part $M_L$ with Lorentzian signature  lies in the region $t>0$ and, regarding $M_L$ as a strongly causal spacetime (with the time orientation provided by $\partial_t$), one can consider its causal 
	boundary\footnote{%For this purpose, we will assume along this discussion that $M_L$ is strongly causal, see   \cite{FlHeSa11}  for the background  on this boundary used here.	
		Notice also that the   change of variable  $t\, \in (0, \infty)\mapsto T$ with  $dT= \sqrt{t}dt$   yields  a simple  metric of the type $-dT^2+g_T$  where $g_T$ varies smoothly with $T$ for $T\neq 0$ and is continuous for  $T=0$.  In particular,  $T$ can be used on $M_L$ and, then,  techniques for the causal boundary as in \cite[Sections 3 and 5.3]{Sanche22} apply.} $\partial_cM_L$ (see   \cite{FlHeSa11}). Each integral curve of 
	$-\partial_t$ is then a past-directed timelike curve $\gamma$ and its chronological future $I^+(\gamma)$ is a  terminal indecomposable future set  TIF (recall that the dual notion of TIP applies to future-directed timelike curves which are inextendible in $M_L$). Then, $I^+(\gamma)$ can be identified with the point of $\HH$ which is limit of $\gamma$. It is straightforward to check that each {\em $\gamma$ is a timelike pregeodesic of $M_L$ which arrives in a finite proper time at $\HH$.}
	
	\item[$(2)$]
	In general, $\partial_cM_L=\partial_c ^-M_L\cup \partial_c ^+ M_L \cup 
	\partial_c ^{\hbox{\tiny{naked}}}M_L$, where  
	$\partial_c ^-M_L$ (resp. $\partial_c ^+M_L$) is the past (resp. future) causal boundary  containing all the TIFs (resp. TIPs), and $\partial_c ^{\hbox{\tiny{naked}}}M_L$ contains all the naked singularities (see \cite{Sanche22})
	%(formally defined as pairs $(P,F)$ of a TIP and TIF which are $S$-related). 
	As for classic relativistic spacetimes, it is natural to assume that $M_L$ is globally hyperbolic\footnote{With more generality, our arguments are trivially extendible to the case of spacetimes conformal to globally hyperbolic spacetimes with timelike boundary (which include, for example,  asymptotically AdS ones) by taking into account the results in \cite{AkFlSa21}.} which turns out to be equivalent to $\partial_c ^{\hbox{\tiny{naked}}}M_L=\emptyset$ \cite[Thm. 3.29]{FlHeSa11}. Summing up, {\em from a global cosmological viewpoint, it is natural to assume that  $M_L$ is globally hyperbolic, and the signature change  occurs in one or two hypersurfaces $\HH^-$ and $\HH^+$, each one contained in $\partial_c ^-M_L$ and  $\partial_c ^+M_L$, respectively.}
	
	\item[$(3)$]
	The counterpart of global hyperbolicity for  Riemannian metrics is metric completeness. In our setting,  the behavior of $g$ in (\ref{e1}) on the Riemannian part $M_R$  (i.e., $t<0$) 
	suggests  that the usual Cauchy boundary   $\partial_{\hbox{\tiny{Cauchy}}} M_R$ for the metric completion should be identified with $\HH$. 
	That is, {\em in a globally hyperbolic framework, the transitions Riemann-Lorentz should occur at hypersurfaces $\HH$ which provide both the Cauchy boundary $\partial_{\hbox{\tiny{Cauchy}}} M_R$  for the Riemannian part $M_R$ and  $\partial_c ^-M_L$ or  $\partial_c ^+M_L$  for $M_L$ (or a connected part of $M_L$)}. Following natural intuitions as in \cite{HarHaw83}, in principle one considers that $M_L$ is connected and there is a single transition (with $\HH$ equal to either $\partial_c ^-M_L$ or $\partial_c ^+M_L$) or two transitions $\HH^-, \HH^+$,  with either one or two connected parts for $M_R$. However, proposals as Penrose's {\em Cosmological Cyclic Cosmology} \cite{Wikipe25} might suggest even an infinite sequence of transitions with infinitely many connected components for $M_L$. 
	
	\item[$(4)$]
	The signature change in (\ref{e1}) corresponds to the degeneracy of the causal cone into a line (the radical of $g$ on $\HH$). However, one can consider a similar degeneration of the dual metric $g^*$ on the cotangent bundle $T^*M$. From the tangent bundle viewpoint, such a signature change  corresponds to the degeneracy of the  cone of $g$ into a hyperplane, namely:
	\begin{equation}\label{e2}
		g=-\frac{1}{t}dt^2+g_{ij}(t,x^1, \dots, x^{n-1})dx^idx^j, \qquad \HH^*=\{t=0\}.
	\end{equation}  
	The integrability  of 
	$1/\sqrt{t}$ permits to obtain a similar 
	conclusion as above for an integral curve $\gamma$ of $-\partial_t$ as well as for the causal boundary\footnote{In spite of the fact that the cones  degenerate to a hyperplane at $\HH^*$, distinct points of $\HH^*$ yield distinct TIFs.		Indeed,  the TIF $F_0$ associated to each $(0,x_0)\in \HH^*$ is $F_0=\{(t,x): 0<t<\left(3|x-x_0|/2\right)^{2/3}\}$.}. In particular, each {\em $\gamma$ is a timelike pregeodesic for the metric in (\ref{e2})  which arrives in a finite proper time at $\HH^*$, and $\HH^*$ is included in the past causal boundary for $g^\ast$.
	}
	\item[$(5)$]
	It is worth stressing that transitions for dual $g^*$ are both mathematically and physically as natural as those for $g$.  Indeed, the  transitions in $g^*$ have the following natural meaning (see \cite{BeLoSa02} for a detailed development and \cite[Sect. 2.1]{BeJaSa20} for further issues): at each point $p$ of $\HH^*$, when one chooses a nonzero form $\omega_p \in T_pM^*$ in the radical of  $g^*_p$, then $g^*_p$ induces an Euclidean scalar product in the kernel of $\omega_p$. Then, up to this choice, $g^*_p$ provides a {\em Galilean} structure on $T_pM$. Analogously, a choice of a nonzero vector $v_p$ in the radical $g_p$ yields also an Euclidean metric in $ v_p^o\subset T_pM^*$ and, thus, a {\em dual Galilean} (or {\em anti-Galilean}) structure\footnote{Anti-Galilean structures are also called {\em Carroll} in the literature. When a Galilean (resp. anti-Galilean) structure is chosen at each point $p$ of a manifold $M$, a Leibnizian (resp. anti-Leibnizian) structure on $M$ is obtained. Leibnizian and anti-Leibnizian structures admit    so many connections that parallelize them as semi-Riemannian metrics do but, as a difference with the latter, a single connection  cannot be selected by imposing torsionless, see \cite{BerSan03}.} on $T_pM$.
	Summing up, {\em the framework of signature-changing metrics should be developed considering transitions of  $g$ and $g^*$ in the same footing,  eventually taking advantage of conclusions on lightcones and other physical issues}. 
		\end{enumerate}

\section{Conclusions}	
 	Smooth metric transitions of signature in both Lorentzian-Euclidean black
	holes and Lorentzian to Riemannian spacetimes, offer interesting
	possibilities to explore. In both of them, there are freely falling
	observers arriving at the transition hypersurface $\mathcal{H}$ in finite
	proper time,  This happens in spite of the numerical evidence found in \cite{CaDeBa24}, thus, the Lorentz-Euclidean black hole model should be revised. 
	
	There is a number of issues which can be addressed directly in the case of
	a Lorentz to Riemannian signature change:
	(1) The existence of a privileged family of freely falling observers
	arriving at finite proper time.
	(2) The  possible   identification of $\mathcal{H}$ with the spatial part of the causal boundary in the Lorentzian region
	$M_L$.
	(3) The extension of the notion of {\em global hyperbolicity} to the
	signature changing metric.
	(4) The fact that signature transitions for the dual metric $g^*$ are
	physically and mathematically as appealing as those for $g$.
	(5) The existence of natural interpretations for the transitions of $g$ and
	$g^*$ in terms of the degeneracy of the lightcones (to a line or a
	hyperplane), eventually with additional geometric structures (Galilean and
	dual Galilean) therein.

	 Thus, the transition hypersurface as in \cite{HarHaw83, KosKri93}, can be regarded not just as a mathematical artifact for boundary conditions but as a physical object with properties testable from the spacetime viewpoint. 
	
\appendix  
\section{The equation \eqref{dottr}   in Gullstrand-Painlev\'e coordinates}\label{appA}
We show that   \eqref{dottr} (and then  \eqref{propertime})  can  be also derived  in Gullstrand-Painlev\'e coordinates. Let $\tilde\varepsilon$  and $\gamma=\gamma (\tau)= \big(\mathscr T(\tau), r(\tau) \big)$  be as in Section \ref{2b},  but now $g_{\tilde \eps}$      as in \eqref{gps},  with $\tilde\varepsilon$ instead of  $\eps_{\rho,\kappa}$. 
In this coordinates, \eqref{E} becomes
\[  E:=-\eps g_{\tilde \eps}(\partial_{\mathscr T}, \dot{\gamma })= \eps \tilde\eps \Lambda \dot{\mathscr T}- \eps \sqrt{\tilde{\eps }(1-\Lambda)} \dot r \]
Taking into account \eqref{privileged} we then  get: 
\begin{align*}
	-\varepsilon &= -\tilde{\eps}\Lambda\dot{\mathscr T}^2 + \dot{r}^2 + 2\sqrt{\tilde{\varepsilon}(1-\Lambda)}\dot{\mathscr T} \dot{r}\\
	& = -\frac{(\eps \tilde \eps \Lambda \dot{\mathscr T})^2}{\eps^2  \tilde \eps  \Lambda} + \dot{r}^2 + 2 \frac{\sqrt{\tilde{\varepsilon}(1-\Lambda)}\dot{r} }{\eps \tilde \eps \Lambda } \eps \tilde \eps \Lambda \dot{\mathscr T} \\
	& = -\frac{(E + \eps \sqrt{\tilde{\eps }(1-\Lambda)}  \dot r)^2}{\eps^2  \tilde \eps  \Lambda} + \dot{r}^2 + 2 \frac{\sqrt{\tilde{\varepsilon}(1-\Lambda)}\dot{r} }{\eps \tilde \eps \Lambda } ( E + \eps \sqrt{\tilde{\eps }(1-\Lambda)}  \dot r) \\
	& = - \frac{E^2}{\eps^2 \tilde \eps \Lambda} + \frac{1-\Lambda}{\Lambda} \dot{r}^2 + \dot{r}^2 \\
	& = - \frac{E^2}{\eps^2 \tilde \eps \Lambda} + \frac{\dot{r}^2}{\Lambda},
\end{align*}
so formula \eqref{dottr}  is obtained.

\section{Christoffel symbols}
	The nonzero Christoffel symbols $\Gamma^r_{ij}$ for the metric $g_{\tilde\eps}$ given in equation  \eqref{gps} (with $\tilde\eps$ replacing $\eps_{\rho,\kappa}$) as calculated using SageMath \cite{SageDe24},   are
	\begin{align}
		&\Gamma^r_{\mathscr T\!\!\mathscr T}=    -\frac{2(2m^2-mr)\tilde\eps (r) -(4m^2r-4mr^2+r^3)\tilde\eps'  (r)}{2r^3}   
		=  \frac{m}{r^2}\beta  (r)+  \frac 12\beta^2  (r) \frac{\tilde\varepsilon'  (r)}{\tilde\varepsilon^2  (r)}  \label{rtt}\\
		&\Gamma^r_{rr}  = 	\frac{mr  \tilde\eps'  (r)- m\tilde\eps (r)}{r^2\tilde\eps (r)} 
		=   - \frac{m}{r^2} +  \frac{m}{r} \frac{\tilde\varepsilon'  (r)}{\tilde\varepsilon  (r)}  \label{rrr}\\
		&\Gamma^r_{\mathscr Tr}   =  - \sqrt{2 m} \;  \frac{  2 m\tilde\varepsilon(r) - (2 mr- r^2) \tilde\varepsilon' (r)  }{2r^2 \sqrt{r \tilde\varepsilon(r)}} =  - \frac{m}{r^2} \sqrt{\frac{2m}{r}}\sqrt{\tilde\varepsilon}- \frac{\beta}{2}\frac{\tilde\varepsilon'}{\tilde\varepsilon^2}\sqrt{\frac{2m}{r}}\sqrt{\tilde\varepsilon}  \label{rtr}\\
		&\Gamma^r_{\theta\theta}=2m-r \label{rthth}\\
		&\Gamma^r_{\varphi\varphi}=(2m-r)\sin^2\theta.\label{rphph}
	\end{align}

	\begin{acknowledgments}
	%	\noindent  We thank an anonymous referee for their   valuable comments. 
		
		\noindent R.B. and E.C. are  partially supported by  PRIN 2022 PNRR {``P2022YFAJH Linear and Nonlinear PDE's: New directions and Application''},  by  MUR under the Programme ``Department of Excellence'' Legge 232/2016  (Grant No. CUP - D93C23000100001) and by  GNAMPA INdAM - Italian National Institute of High Mathematics.
		
		\noindent M.S. is partially supported by the
		project PID2020-116126GB-I00 funded by MCIN/ AEI /10.13039/501100011033  and by the
		framework of IMAG-Mar\'{\i}a de Maeztu grant CEX2020-001105-M funded 
		by MCIN/AEI/10.13039/50110001103.
		
		\noindent The data that support the findings of this article are openly available \cite{SageDe24}. 
	\end{acknowledgments}

%\bibliography{../../mybib}
	

	\end{document}